\documentclass{iopart}
\usepackage{iopams}
\usepackage{graphicx}
\usepackage[colorlinks,linkcolor=blue]{hyperref}

\def\>{\rangle}
\def\<{\langle}
\def\({\left(}
\def\){\right)}

\newcommand{\eqref}[1]{Eq.~(\ref{#1})}

\renewcommand{\vec}[1]{\mathbf{#1}}

\def\th{^{\rm th}}

\newtheorem{theorem}{Theorem}

\newtheorem{proposition}[theorem]{Proposition}

\newenvironment{proof}[1][Proof]{\noindent\textbf{#1.} }{\ \rule{0.5em}{0.5em}}

\begin{document}

\title{Weighing matrices and optical quantum computing}

\author{Steven T. Flammia$^1$ and Simone Severini$^2$}
\address{$^1$Perimeter Institute for Theoretical Physics, 31 Caroline St.\ N, Waterloo, Ontario, N2L 2Y5 Canada}
\address{$^2$Institute for Quantum Computing and Department of Combinatorics \& Optimization, University of Waterloo, 200 University Ave. W, Waterloo, Ontario, N2L 3G1 Canada }
\ead{\mailto{sflammia@perimeterinstitute.ca}, \mailto{simoseve@gmail.com}}

\begin{abstract}
Quantum computation in the one-way model requires the preparation of certain resource states known as cluster states.  We describe how the construction of continuous-variable cluster states for optical quantum computing relate to the existence of certain families of matrices.  The relevant matrices are known as weighing matrices, with a few additional constraints.  We prove some results regarding the structure of these matrices, and their associated graphs.  
\end{abstract}


%------------------------------------------------------------------------------------------------------------%
\section{Introduction}\label{S:intro}
%------------------------------------------------------------------------------------------------------------%

In the standard model of quantum computing \cite{Nielsen2000}, the computation proceeds by coherently performing unitary dynamics on a simple initial state of a quantum system before being read out by a sequence of local measurements.  By contrast, the one-way model of quantum computing \cite{Raussendorf2001} eliminates the need for coherent quantum evolution, substituting instead a highly entangled, but easily prepared initial state as the resource for the computation.  The typical resource state used is the {\it cluster state\/} \cite{Briegel2001}, whose preparation consists of, for example, coupling a system of qubits with a nearest-neighbor Ising interaction.  The computation proceeds by a sequence of local measurements on the cluster state and the basis chosen at each step determines the computation.  In this way, a cluster state is like a quantum breadboard on which the quantum circuit is inscribed via measurements.

Although cluster states were originally defined using qubits, they generalize to $d$-level quantum systems \cite{Zhou2003} and continuous variables \cite{Zhang2006, Menicucci2006}.  Recently, Ref. \cite{Menicucci2008} proposed a method for efficiently generating continuous-variable cluster states (CVCS).  The method relies on pumping an optical cavity containing a special nonlinear medium at certain carefully chosen frequencies.  The resulting CVCS is encoded in the quadratures of the photons that emerge from the cavity.

The method of CVCS generation proposed in \cite{Menicucci2008} required, on grounds of experimental and theoretical tractability, the existence of certain families of matrices.  The required matrices are {\it weighing matrices\/}, defined as matrices $W$ whose entries are in $\{0,\pm1\}$ and satisfy $WW^T = k I$, where $k$ is called the \emph{weight}.  The weighing matrices were obliged to have \emph{Hankel} form, meaning that the skew-diagonals of the matrix were all constant, and any skew-diagonal with support on the main diagonal must be identically 0.  A matrix with vanishing main diagonal is said to be \emph{hollow}.  As we show later, when viewing the matrices as adjacency matrices for an associated graph, this is implied when the graph is connected and bipartite.

In this article, we investigate the mathematical structure behind the construction in \cite{Menicucci2008} and begin to classify the associated matrices.  The discussion is organized as follows.  In \sref{S:physics} we give a brief glimpse of the physics which motivates the restrictions given in the previous paragraph.  We begin the mathematical discussion in \sref{S:ag}, where we discuss anticirculant graphs, and \sref{S:awm}, where we discuss anticirculant weighing matrices, a special case of the Hankel weighing matrices.  In \sref{S:2reg}, we complete the rather trivial classification of the case of weight 2, while \sref{S:4reg} provides several examples about the case of weight 4.  We discuss a number of open problems in \sref{S:openproblems}, and conclude in \sref{S:conc}.

%------------------------------------------------------------------------------------------------------------%
\section{Quantum computing in the frequency comb}\label{S:physics}
%------------------------------------------------------------------------------------------------------------%

In this section we outline the physics that naturally leads us to consider Hankel hollow weighing matrices.  For a more detailed discussion of the physics, see \cite{Menicucci2008, Menicucci2007, Zaidi2008, Flammia2008}.

The resonant frequencies of an optical cavity are defined classically as the modes which constructively interfere inside the cavity.  Because they are evenly spaced, they form a so-called a {\it frequency comb\/}.  By placing a nonlinear medium inside the cavity, pump photons with frequency $\omega_p$ can downconvert into an entangled pair of photons, so long as the new photons' frequencies satisfy an energy conservation constraint,
\begin{equation}\label{E:photonEC}
	\omega_p=\omega_m+\omega_n \,,
\end{equation}
where $\omega_m$ and $\omega_n$ are the $m\th$ and $n\th$ resonant frequencies of the cavity.  The range over which the downconversion can occur is limited by the phasematching bandwidth of the cavity.

More generally, multiple frequencies can be simultaneously phasematched inside the cavity, and upconversion is allowed as well as downconversion.  The total Hamiltonian is then idealized by the following
\begin{equation}\label{E:Hamiltonian}
	\mathcal{H} = i \hbar \kappa \sum_{p \in P}\ \sum_{m+n=p}\! M_{mn} (\hat a^\dag_m \hat a^\dag_n - \hat a^{\phantom\dag}_m \hat a^{\phantom\dag}_n) ,
\end{equation}
where $P$ is the (discrete) spectrum of the polychromatic pump, $\kappa$ is a global coupling strength, $\hat a^\dag_n$ is the creation operator for the $n\th$ mode, and $M_{mn}$ are matrix elements of a symmetric matrix.  In units of the fundamental frequency of the cavity, the restriction of the inner sum to terms with $m+n=p$ enforces photon energy conservation, and the magnitude of $M_{mn}$ determines the strength of the coupling in units of $\hbar\kappa$, while the sign determines whether the photons are upconverted or downconverted.  If all of the coupling strengths are equal, then $M$ is a Hankel matrix whose elements are either $0$ or $\pm c$, where $c$ is a constant.  Experimental simplicity further demands that there is no single-mode squeezing and hence we require that the elements along the main diagonal of $M$ are all zero, i.e.~$M$ must be a hollow matrix.

The Hamiltonian \eref{E:Hamiltonian} can be used to create continuous-variable cluster states, defined as Gaussian states satisfying the relation
\begin{equation}\label{E:CVCSdef}
	\vec p - A \vec q \rightarrow \vec 0\;,
\end{equation}
where $\vec p$ and $\vec q$ are vectors of amplitude and phase quadratures for the $k\th$ cavity mode, $q_k=\hat a^\dag_k+\hat a^{\phantom \dag}_k$ and $p_k=i(\hat a_k^\dag-\hat a^{\phantom \dag}_k)$, $A$ is the symmetric weighted adjacency matrix of the graph of the CVCS, and the arrow denotes the limit of large squeezing.  Reference \cite{Menicucci2007} first proved a nonunique relationship between the adjacency matrix $A$, which describes the cluster state, and the adjacency matrix $M$, which describes the coupling between the modes inside the cavity.  The general relationship is somewhat complicated, so references \cite{Zaidi2008, Menicucci2008, Flammia2008} introduced the simplifying ansatz that 
\begin{equation}\label{E:Asquared}
	AA^T = A^2 = 1,
\end{equation}
from which it follows \cite{Zaidi2008, Flammia2008} that (up to a trivial relabeling)
\begin{equation}\label{E:GequalsA}
	M=A .
\end{equation}
In short, assuming $A$ is an orthogonal matrix implies the graph describing the couplings between photons in the cavity ($M$) is \emph{identical} to the graph describing the cluster state ($A$), which is a vast simplification over the general relationship.  The derivation requires that $M$ be the adjacency matrix of a bipartite graph, but as we will see, one can derive this from the Hankel and hollow constraints.

We now see how the twin demands of experimental and theoretical simplicity for creating CVCS in an optical cavity naturally lead us to consider Hankel hollow orthogonal matrices, all of whose non-zero elements are either $\pm c$, for some constant~$c$.  We can of course consider renormalized matrices instead by dividing out $c$, so that all entries are $\pm1$; after finding these renormalized matrices, we can then reintroduce the constant afterwards so that \eref{E:Asquared} holds.  Thus, we see that finding and classifying Hankel hollow weighing matrices is our primary interest.

For completeness, we mention one further elaboration that is possible for this scheme.  If the photons' polarization, transverse, or spatial degrees of freedom are taken into account, then the Hamiltonian \eref{E:Hamiltonian} can be modified by adding additional indices for these modes.  If these additional modes are all frequency-degenerate, then an additional sum over these degenerate modes appears in \eref{E:Hamiltonian}, which allows for the symmetric intercoupling between the degenerate degrees of freedom.  If such an interaction could be simultaneously phasematched by the nonlinear medium inside the cavity, then one is naturally lead to consider \emph{block}-Hankel matrices, where the size of the blocks is equal to the number of degenerate degrees of freedom.  We leave the consideration of block-Hankel hollow weighing matrices to future work.

%------------------------------------------------------------------------------------------------------------%
\section{Anticirculant graphs}\label{S:ag}
%------------------------------------------------------------------------------------------------------------%

In this section we introduce the definition of anticirculant graphs and study some of their first properties. For the sake of self-containedness, we will define here all the graph-theoretic concepts we require. For further background on theory of graphs the interested reader is referred to the book by Diestel \cite{Diestel2000}. Our reference about permutations and finite groups is Cameron \cite{Cameron1999}.

We will work with simple graphs. A (simple) \emph{graph} $G=(V,E)$ is an ordered pair of sets defined as follows: $V(G)$ is a non-empty set, whose elements are called \emph{vertices}; $E(G)$ is a non-empty set of unordered pairs of vertices, whose elements are called \emph{edges}. An edge of the form $\{i,i\}$ is called a \emph{loop}. Two vertices $i$ and $j$ are said to be \emph{adjacent} if $\{i,j\}\in E(G)$. Often, this is simply denoted by writing $ij$. A graph is \emph{loopless} if it is has no loops. The \emph{adjacency matrix} of a graph $G$ is denoted by $A(G)$ and defined by 
\[
[A(G)]_{i,j}:=\left\{ 
\begin{tabular}{ll}
$1,$ & if $ij\in E(G);$ \\ 
$0,$ & if $ij\notin E(G).$%
\end{tabular}%
\right. 
\]%
An adjacency matrix of a loopless graph is said to be \emph{hollow}.  The \emph{degree} of a vertex $i$ in a graph $G$ is the number of edges incident with the vertex, that is, $\#\{j:ij\in E(G)\}$. A graph is \emph{regular} if each of its vertices has the same degree. We say $d$\emph{-regular} to specify that each vertex has degree $d$. A square matrix is \emph{Hankel} if it has constant skew-diagonals.  These are diagonals that traverse the matrix from North-East to South-West. (For this reason, these are also called \emph{antidiagonals}). Notice that an $n\times n$ matrix has $2n-1$ skew-diagonals. An $n\times n$ Hankel matrix is said to be \emph{anticirculant }(or, equivalently, \emph{skew-circulant} or \emph{backcirculant}) if the $i$-th and $(i+n)$-th skew diagonals are equal. A \emph{permutation} of \emph{length} $n$ is a bijection $\pi
:[n]\longrightarrow \lbrack n]$, where $[n]=\{1,2,...,n\}$. A \emph{%
permutation matrix} of \emph{dimension }$n$ is a $n\times n$ matrix $P$ with
the following two properties: $[P]_{i,j}\in \{0,1\}$ for every $i,j$; $P$
has a unique $1$ in each row and in each column. A permutation $\pi $ is
said to \emph{induce} a permutation matrix $P$ if $[P]_{i,\pi (i)}=1$ for
every $i$.

Two graphs $G$ and $H$ are \emph{isomorphic} if there is a permutation
matrix $P$ such that $A(G)=PA(H)P^{-1}$. We write $G\cong H$ to denote that
graphs $G$ and $H$ are isomorphic. A graph $G$ is \emph{bipartite} if there
is a permutation matrix $P$ such that 
\begin{equation}
PA(G)P^{-1}=\left( 
\begin{array}{cc}
0 & M \\ 
M^{T} & 0%
\end{array}%
\right) ,  \label{ma}
\end{equation}%
for some matrix $M$. If $M=J_{n}$, where $J_{n}$ is the all-ones matrix of
size $n$, then $G\cong K_{n,n}$, the \emph{complete bipartite graph}. A
graph is said to be \emph{anticirculant} if it is isomorphic to a graph
whose adjacency matrix is anticirculant. Let us denote by $\mathcal{L}_{n,d}$
and by $\mathcal{L}_{n}$ the set of all $d$-regular anticirculant graphs on $%
n$ vertices and the set of anticirculant graphs on $n$ vertices,
respectively. Notice that $\#\mathcal{L}_{n,d}={n \choose d}$. Therefore, 
\[
\#\mathcal{L}_{n}=\sum\nolimits_{d=1}^{n-1}\#\mathcal{L}_{n,d}=\sum%
\nolimits_{d=1}^{n-1}{n \choose d}=2^{n}-2. 
\]%
Given $G\in \mathcal{L}_{n}$, the ordered set $S(G)=(s:[A(G)]_{1,s}=1)$ is
said to be the \emph{symbol} of $G$, in analogy with the terminology used
for circulant graphs. Indeed, $S(G)$ specifies $G$ completely. In cycle
notation, each element $s\in S(G)$ is associated to a permutation of the form%
\begin{equation}
\pi _{s}=(1,s)(2,s-1)\cdots (s+1,n)(s+2,n-1)\cdots ,  \label{per}
\end{equation}%
where the subtraction is modulo $n$. In line notation, $\pi
_{s}=(s)(s-1)(s-2)\cdots (n)(n-1)\cdots (s+1)$. On the basis of the
definitions, we can collect the facts below:

\begin{proposition}
\label{ACGraphs}Let $G\in \mathcal{L}_{n}$ be a graph. The following
statements hold true:

\begin{enumerate}
\item The graph $G$ is $d$-regular with $d=\#S(G)$.

\item If $G$ is loopless then $n$ is even.

\item If $G$ is loopless then each $s\in S(G)$ is even.

\item The graph $G$ is bipartite.
\end{enumerate}
\end{proposition}

\begin{proof}
In order:
\begin{enumerate}
\item Let $\pi _{s}$ be the permutation associated to the element $s\in S(G)$. Then $\pi _{s}$ induces a symmetric permutation matrix $P_{s}$.  Since each $P_{s}$ is defined by $s$, we have $[P_{s}]_{i,j}=1$ if and only if $ [P_{t}]_{i,j}=0$ for all $t\not=s$.  Consequently, $A(G)=\sum_{s\in S(G)}P_{s}$, and $G$ is a regular graph with degree $d=\#S$.

\item If the size of $P_{s}$ is odd, then $[P_{s}]_{i,i}=1$ for some $i$. This is because $\pi _{s}$ is necessarily an involution of odd length, and for this reason it has at least one fixed point. (Recall that an \emph{involution} is a permutation that is its own inverse.)

\item Suppose on the contrary that $s=2k+1$, for some $k\geq 1$. Then, $[A(G)]_{1,s}=[A(G)]_{1,2k+1}=[A(G)]_{k+1,k+1}=1$, contradicting the hypothesis of $G$ being loopless.

\item By the form of each $\pi _{s}$ (see Eq. \ref{per}), it is sufficient to observe that vertices labeled by odd numbers are adjacent only to vertices labeled by even numbers and \emph{viz}. 
\end{enumerate}
\end{proof}

\bigskip

At this stage, we are ready to prove that the members of $\mathcal{L}_{n}$
are Cayley graphs. The \emph{Cayley graph} $G=G(\Gamma ,T)$ of a group $%
\Gamma $ \emph{w.r.t.} the set $T\subseteq \Gamma $ is the graph in which $%
V(G)=\{\Gamma \}$ and $\{g,h\}\in E(G)$, if there is $s\in T$ such that $%
gs=h $. The dihedral group $D_{2k}$ of order $2k$ is the group of symmetries
of a regular $k$-gon. The dihedral group $D_{2k}$ is nonabelian and
presented as 
\[
D_{2k}=\langle s,t:t^{k}=e,s^{2}=e,sts=t^{-1}\rangle , 
\]%
where $t$ is a rotation and $s$ is a reflection. The elements of $D_{2k}$
are $k$ rotations $t^{0},t,t^{2},...,t^{k-1}$ and $k$ reflections $%
s,st,st^{2},...,st^{k-1}$. A graph $G$ is said to be \emph{connected} if
there is no permutation matrix $P$ such that $PA(G)P^{-1}=\bigoplus_{i}M_{i}$%
, for some matrices $M_{i}$. If a Cayley graph of a group $\Gamma$ is not
connected then it is the disjoint union of isomorphic Cayley graphs of a
subgroup of $\Gamma$, with each isomorphic component corresponding to a coset of the subgroup.

\begin{proposition}
\label{propore}Let $G$ be a loopless anticirculant graph. Then $G$ is a
Cayley graph of the dihedral group with respect to a set of reflections.
\end{proposition}

\begin{proof}
The adjacency matrix of a Cayley graph $G=G(\Gamma ,T)$ can be written as $%
A(G)=\sum_{t\in T}\rho _{\rm{reg}}(t)$, where $\rho _{\rm{reg}}$ is the regular
permutation representation of $\Gamma $. When the order of $\Gamma $ is $n$,
this is an homomorphism of the form $\rho _{\rm{reg}}:\Gamma \longrightarrow
\Sigma _{n}$, where $\Sigma _{n}$ is the set of all $n\times n$ permutation
matrices. Each $\rho _{\rm{reg}}(t)$ describes the action of the group element $t$
on the set $\{1,2,...,n\}$. Let $G\in \mathcal{L}_{n}$ be a loopless graph.
By Proposition \ref{ACGraphs}, we need to take $n=2k$. Let us label the
lines (rows and columns) of $A(G)$ with the group elements of $D_{n}$, of
order $n=2k$, in the following order:%
\[
\begin{tabular}{lllllllll}
$t^{0}=e$ & $s$ & $t$ & $st$ & $t^{2}$ & $st^{2}$ & $\cdots $ & $t^{k-1}$ & $%
st^{k-1}$%
\end{tabular}%
, 
\]
where $e$ is the identity element of the group. Consider $\rho
_{\rm{reg}}(st^{l}) $, where $1\leq l\leq k$. By applying the generating
relations of $D_{n}$, it is straighforward to verify that $\rho
_{\rm{reg}}(s)=P_{2}$, and $\rho _{\rm{reg}}(st^{l})=P_{2+2l}$, when $1\leq l\leq k-1$%
. The elements of $S(G)$ are then associated to reflections in $D_{2k}$: 
\[
\begin{tabular}{lllll}
$2\longrightarrow s$ & $4\longrightarrow st$ & $6\longrightarrow st^{2}$ & $%
\cdots $ & $2k\longrightarrow st^{k-1}$%
\end{tabular}%
. 
\]%
This concludes the proof of the proposition.
\end{proof}

\bigskip

An example is useful to clarify this result. Let $\mathbb{Z}_{n}$ be the
additive group of integers modulo $n$. The Cayley graph $X(\mathbb{Z}%
_{n},\{1,n-1\})$ is also called the $n$\emph{-cycle}. This is the unique
connected $2$-regular graph up to isomorphism. Let $G\in \mathcal{L}_{6,2}$\
be the graph with adjacency matrix%
\[
A(G)=\rho _{\rm{reg}}(s)+\rho _{\rm{reg}}(st)=\left( 
\begin{array}{cccccc}
0 & 1 & 0 & 1 & 0 & 0 \\ 
1 & 0 & 1 & 0 & 0 & 0 \\ 
0 & 1 & 0 & 0 & 0 & 1 \\ 
1 & 0 & 0 & 0 & 1 & 0 \\ 
0 & 0 & 0 & 1 & 0 & 1 \\ 
0 & 0 & 1 & 0 & 1 & 0%
\end{array}%
\right) . 
\]%
where $s=(1,2)(3,6)(4,5)$ and $t=(1,4)(2,3)(5,6)$. Then $s^{2}=st^{2}=e$ and 
$\langle s,st\rangle =D_{6}$. Additionally, it is immediate to see that $%
G\cong X(\mathbb{Z}_{6},\{1,5\})$.

%------------------------------------------------------------------------------------------------------------%
\section{Anticirculant weighing matrices}\label{S:awm}
%------------------------------------------------------------------------------------------------------------%

In this section we will highlight the interplay between weighing matrices
and anticirculant graphs supporting orthogonal matrices.

A Weighing matrix $W$ of \emph{order} $n$ and \emph{weight} $k$, denoted by $%
W(n,k)$, is a square $n \times n$ matrix with entries $[W]_{i,j}\in
\{-1,0,1\}$, satisfying $WW^{T}=kI$, where $I$ is the identity matrix. When $k=n$, then $W(n,k)$ is said to be a \emph{Hadamard
matrix.} It is generally recognized that weighing matrices were first
discussed by Frank Yates in 1935 \cite{Yates1935}, while Hadamard matrices
were introduced by James Sylvester and Jacques Hadamard during the second
half of the nineteenth century (see \cite{Horadam2007}). Geramita and
Seberry \cite{Geramita1979}, and Koukouvinos and Seberry \cite%
{Koukouvinos1997} are general surveys on this topic, as well as applications.

We are specifically interested in circulant weighing matrices. A \emph{%
circulant weighing matrix} of order $n$ and weight $k$ is a $W(n,k)$ which
is also a circulant matrix. In analogy with the previous section, we define the
ordered set $T(W(n,k))=\{s:[A(CW(n,k))]_{i,s}=1\}$ to be the symbol of $%
CW(n,k)$. Arasu and Seberry overview the subject in \cite{Arasu1998}.

We shall give particular importance to graphs associated to weighing
matrices. The \emph{graph of a }(real symmetric) \emph{matrix} $M$, denoted
by $G(M)$, is the graph defined by 
\[
\lbrack A(G(M))]_{i,j}:=\left\{ 
\begin{tabular}{ll}
$1,$ & if $[M]_{i,j}\neq 0;$ \\ 
$0,$ & if $[M]_{i,j}=0.$%
\end{tabular}%
\right. 
\]%
The graph $G(CW(n,k))$ is a Cayley graph of the cyclic group $\mathbb{Z}_{n}$
with respect to the set $T(W(n,k))$. For instance, let us look at the matrix 
\begin{equation}
CW(6,4)=\left( 
\begin{array}{rrrrrr}
0 & 1 & 1 & 0 & 1 & -1 \\ 
-1 & 0 & 1 & 1 & 0 & 1 \\ 
1 & -1 & 0 & 1 & 1 & 0 \\ 
0 & 1 & -1 & 0 & 1 & 1 \\ 
1 & 0 & 1 & -1 & 0 & 1 \\ 
1 & 1 & 0 & 1 & -1 & 0%
\end{array}%
\right)  \label{ci6}
\end{equation}%
The graph $G(CW(6,4))$ is the Cayley graph of $\mathbb{Z}_{6}$ with respect
to the set $\{1,2,4,5\}$. This is illustrated in Figure \ref{circ1}.

\begin{figure}
\begin{center}
\includegraphics[height=4cm]{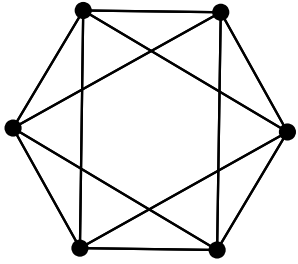}
\caption{The Cayley graph $X(\mathbb{Z}_{6},\{1,2,4,5\})$.}
\label{circ1}
\end{center}
\end{figure}

The following statement is one of the main tools in our discussion. For the
purposes of this paper, it is sufficient to focus on matrices of even order.

\begin{proposition}
\label{Reduction}For each circulant weighing matrix $M=CW(n,k)$ of order $%
n=2l$, with $l\geq 1$, there is a weighing matrix $U$ such that $G(U)\in 
\mathcal{L}_{n,k}$. The graph $G(U)$ is loopless if each element of $%
T(W(n,k))$ is odd.
\end{proposition}

\begin{proof}
Consider the permutation 
\begin{equation}
\pi =\left( 1,n\right) ,(2,n-1),...,(l,l+1)  \label{eqp}
\end{equation}%
on a set $[n]$. Let $P_{\pi }$ be the permutation matrix induced by $\pi $.
Labeling the rows and the columns of $P_{\pi }$ with the elements of $[n]$
in the lexicographic order, we can write 
\[
P_{\pi }=\left( 
\begin{array}{cccc}
0 & \cdots & 0 & 1 \\ 
\vdots & 0 & 1 & 0 \\ 
0 &  & \ddots & \vdots \\ 
1 & 0 & \cdots & 0%
\end{array}%
\right) . 
\]%
Now let $U = MP_{\pi }$. Since $P_\pi$ and $M$ are orthogonal matrices, $U$
is also orthogonal. By the action of the permutation $\pi $ and the fact that $M$
is circulant, the graph $G(U)\in \mathcal{L}_{n,k}$ and $S(G(U))=%
\{s=n-[A(CW(n,k))]_{i,t}+1\}$, where $t\in T(W(n,k))$. Notice that $G(U)$ can
have self-loops if we do not impose further restrictions. On the other hand $%
G(U)$ is loopless if each $s\in S(G(U))$ is even (Proposition \ref{ACGraphs}%
). Such a property is satisfied only if each element of $T(W(n,k))$ is odd.
\end{proof}

\bigskip

It is important to remark that $P_{\pi }$ acts on columns and hence it does
not generally preserve isomorphism of graphs. Proposition \ref{Reduction}
prompts us to some definitions. An \emph{anticirculant weighing} \emph{%
matrix }of order $n$ and weight $k$, denoted by $AW(n,k)$, is a $W(n,k)$
which is also anticirculant. A matrix $AW(n,k)$ is a special kind of Hankel
weighing matrix. A \emph{Hankel weighing} \emph{matrix }of order $n$ and
weight $k$, denoted by $M=HW(n,k)$, is a $W(n,k)$ such that $G(M)$ is
anticirculant. While it is obvious that there exists an $AW(n,k)$ if and
only if there exists a $CW(n,k)$, we need to take into account the following
counterexample:

\begin{proposition}
\label{Inequivalence}Not every $HW(n,k)$ is an $AW(n,k)$.
\end{proposition}

\begin{proof}
Let $P_{\pi }$ be the permutation matrix as in the proof of Proposition \ref%
{Reduction}. The matrix 
\begin{equation}
U=\left( 
\begin{array}{cc}
P_{\pi } & P_{\pi } \\ 
P_{\pi } & -P_{\pi }%
\end{array}%
\right)  \label{2Reg}
\end{equation}%
is a $HW(n,2)$, for every $n$, since 
\[
UU^{T}=U^{2}=\left( 
\begin{array}{cc}
2P_{\pi }^{2} & 0 \\ 
0 & 2P_{\pi }^{2}%
\end{array}%
\right) =\left( 
\begin{array}{cc}
2I & 0 \\ 
0 & 2I%
\end{array}%
\right) , 
\]%
because $P_{\pi }=P_{\pi }^{T}$.
\end{proof}

\bigskip

Proposition \ref{Inequivalence} seems to indicate that classification theorems for circulant weighing matrices help only partially when attempting to classify Hankel weighing matrices.

In the next two sections, our classification begins by considering graphs of small degree.  This corresponds physically to the number of pump beams required to build a given CVCS.  If a graph is $r$-regular, it requires $2r-1$ different pump beams to build the associated CVCS, since each of the $2r-1$ bands of the adjacency matrix requires a different pump frequency.  Thus, we are primarily interested in graphs that have an interesting topology, but also have small values of $r$, since greater values would have greater experimental complexity.  Cases of small $r$ are also more tractable theoretically, which is another reason to focus on them.  

%------------------------------------------------------------------------------------------------------------%
\section{2-regular graphs}\label{S:2reg}
%------------------------------------------------------------------------------------------------------------%

The unique -- up to isomorphism -- connected $2$-regular graph on $n$
vertices is the $n$\emph{-cycle}, $C_{n}$. A graph $G$ is said to be the 
\emph{disjoint union} of graphs $H_{1},...,H_{l}$ if there is a permutation
matrix $P$ such that $PA(G)P^{-1}=\bigoplus_{i=1}^{l}A(H_{i})$. In such a
case, we write $G\cong \biguplus\nolimits_{i=1}^{l}H_{i}$. Equivalently, let 
$G$ and $H$ be graph such that $V(G)\cap V(H)=\emptyset $. The set of
vertices of $G\uplus H$ is $V(G)\cup V(H)$ and the set of edges, $E(G)\cup
E(G)$. All useful information about graphs of matrices $HW(n,2)$ is stated
as follows:

\begin{proposition}
\label{prod}Let $D_{n}$ be the dihedral group of order $n=2k$. Let $T=\{st^{i},st^{j}\}\subset D_{n}$, where $1\leq i,j\leq k$, be a set of reflections. There is a unitary matrix $U$ such that $X(D_{n},T)\cong G(U)$ if and only if $(i-j)\bmod k=(j-i)\bmod k$. Moreover, if this is the case then $k$ is even and therefore $n$ is a multiple of $4$. It follows that the graph $G\in \mathcal{L}_{n,2}$ and $G$ is the disjoint union of $k/2$ copies of $C_4$.
\end{proposition}

\begin{proof}
It is known that given a group $\Gamma $ and a set $T=\{g,h\}\subset \Gamma $, there is a unitary matrix $U$ such that $A(X(\Gamma ,T))=A(G(U))$ if and only if $gh^{-1}=g^{-1}h$ (see, \emph{e.g.}, \cite{Severini2003, Severini2008}). From this, $st^i(st^j)^{-1} = (st^i)^{-1}st^j$. Since $st^l$ is a reflection, for every $1\leq l\leq k$, we can write $(st^l)^{-1}=st^l$. Thus $st^{i}st^{j}=st^{j}st^{i}$, that is, $st^{i}$ and $st^{j}$ commute. Since $st^{l}s=t^{k-l}$, for every $1\leq l\leq k$, it follows that $t^{(k-i+j) \bmod k}=t^{(k-j+i)\bmod k}$. This proves that $(i-j)\bmod k=(j-i)\bmod k$, for $st^i$ and $st^j$. Now, without loss of generality, assume that $j>i$. Note that $(i-j) \bmod k=k-j+i$. So, we need that $j-i=k-j+i,$ which implies $j-i=k/2$. It follows that $U$ has the form described in Eq. (\ref{2Reg}), up to isomorphism.
\end{proof}

\bigskip

Figure \ref{D8} illustrates the graph $X(D_{8},\{st,st^{3}\})$ associated to
the matrix $HW(8,2)$.  Since in general only a disjoint union of squares can be achieved with 2-regular graphs, and these graphs have only limited interest from a quantum computing perspective, we need to consider larger graphs.  A proposal for creating such graphs was given in Ref.~\cite{Zaidi2008}.
%\[
%HW(8,2)=\left( 
%\begin{array}{rrrrrrrr}
%0 & 0 & 0 & 1 & 0 & 0 & 0 & 1 \\ 
%0 & 0 & 1 & 0 & 0 & 0 & 1 & 0 \\ 
%0 & 1 & 0 & 0 & 0 & 1 & 0 & 0 \\ 
%1 & 0 & 0 & 0 & 1 & 0 & 0 & 0 \\ 
%0 & 0 & 0 & 1 & 0 & 0 & 0 & -1 \\ 
%0 & 0 & 1 & 0 & 0 & 0 & -1 & 0 \\ 
%0 & 1 & 0 & 0 & 0 & -1 & 0 & 0 \\ 
%1 & 0 & 0 & 0 & -1 & 0 & 0 & 0%
%\end{array}%
%\right) . 
%\]%

\begin{figure}
\begin{center}
\includegraphics[height=2.5cm]{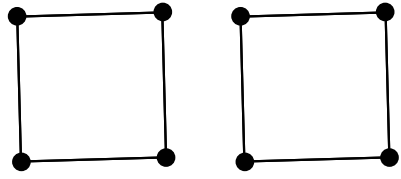}
\caption{The graph $X(D_{8},\{st,st^{3}\}) \protect\cong C_{4}\uplus C_{4}$.}
\label{D8}
\end{center}
\end{figure}

%------------------------------------------------------------------------------------------------------------%
\section{4-regular graphs}\label{S:4reg}
%------------------------------------------------------------------------------------------------------------%

In this section we consider several constructions of 4-regular graphs of anticirculant weighing matrices.  Unfortunately, we are unable to prove that this set of examples is exhaustive, so a complete characterization of the degree 4 case remains open.

The CVCS graphs in this section can all be implemented with an experimental setup requiring no more than 7 pump beams.

%------------------------------------------------------------------------------------------------------------%
\subsection{Graphs from \texorpdfstring{$AW(4,4)$}{AW(4,4)}}\label{AW(4,4)}
%------------------------------------------------------------------------------------------------------------%

The matrix
\[
M_{4,4}=AW(4,4)=\left( 
\begin{array}{rrrr}
1 & 1 & 1 & -1 \\ 
1 & 1 & -1 & 1 \\ 
1 & -1 & 1 & 1 \\ 
-1 & 1 & 1 & 1%
\end{array}%
\right) 
\]%
is a Hadamard matrix of order $4$. Since the entries of a Hadamard matrix
are all nonzero, $G(M_{4,4})\cong K_{4}^{+}$, where $K_{n}^{+}$ denotes the
complete graph on $n$ vertices with a self-loop at each vertex. We can
construct other anticirculant weighing matrices of higher order from $M_{4,4}
$. Given that we are interested in anticirculant matrices whose graph is
without self-loops, the first row of these matrices must have nonzero
entries only at even positions. Let $\mathcal{M}_{4,4}$ be a family of
matrices obtained from $M_{4,4}$ be adding new lines (\emph{i.e.}, rows and
columns) between the rows and the columns of $M_{4,4}$. If we add the same
number of lines between each column (resp. row), we obtain matrices of order 
$4k$, for $k\geq 1$. For small orders, the first rows of these matrices are%
\[
\begin{tabular}{l}
$0,1,0,1,0,1,0,-1$ \\ 
$0,0,1,0,0,1,0,0,1,0,0,-1$ \\ 
$0,0,0,1,0,0,0,1,0,0,0,1,0,0,0,-1$ \\ 
$0,0,0,0,1,0,0,0,0,1,0,0,0,0,1,0,0,0,0,-1$%
\end{tabular}%
.
\]%
Clearly, for each matrix $M$ constructed in this way, we have $G(M)\in L_{n}$%
. When $n=4k$, with $k$ odd, the graph $G(M)$ has self-loops. For this
reason, we consider only the case $n=4k$, with $k$ even. In other words, we
consider the case $n=8l$, for any $l\geq 1$. Let $S(M)=%
\{s_{1},s_{2},s_{3},s_{4}\}$. By the above construction $s_{i+1}-s_i = d+1$ for $i=1,2,3$, and $n-s_{4}+2s_{1}=s_{2}$. This last condition provides that the
number of columns between $s_{4}$ and $s_{1}$ is the same as one between $%
s_{2}$ and $s_{1}$ when looking at the matrix as wrapped on a torus, by
gluing together the first and the last column. In this way, we can define a
unitary matrix $U$ such that $G(U)\cong G(M)$. Let us denote by $%
U_{1},...,U_{n}$ the rows of $U$. The nonzero entries in $U_{1}$ are exactly 
$U_{1,s_{1}},U_{1,s_{2}},U_{1,s_{3}}$ and $U_{1,s_{4}}$. In particular, $%
U_{1,s_{1}}=U_{1,s_{2}}=U_{1,s_{3}}=1$ and $U_{1,s_{4}}=-1$. Suppose $%
s_{i+1}-s_{i}-1=d$ for all $i=1,2,3$ and $n-s_{4}+2s_{1}=s_{2}$. Thus the
inner product between the rows $U_{1}$ and $U_{s_{2}-s_{1}}$ is%
\begin{eqnarray*}
\langle U_{1},U_{s_{2}-s_{1}}\rangle 
&=&U_{1,s_{1}}U_{l,s_{1}}+U_{1,s_{2}}U_{l,s_{2}}+U_{1,s_{3}}U_{l,s_{3}}+U_{1,s_{1}}U_{l,s_{4}}
\\
&=&U_{1,s_{1}}U_{1,s_{2}}+U_{1,s_{2}}U_{1,s_{3}}+U_{1,s_{3}}U_{1,s_{4}}+U_{1,s_{4}}U_{1,s_{1}}
\\
&=&1+1-1-1 \\
&=&0.
\end{eqnarray*}%
The \emph{distance} $d$ guarantees that each row (resp. column) of $U$ has
exactly $4$ nonzero entries contributing to the zero inner product with
other $3$ rows (resp. columns). Orthogonality is guaranteed since no nonzero
entries contribute to the inner product with all remaining rows (resp.
columns). The numbers $s_{1},...,s_{4}$ must be even and $d=(n-4)/4$ must be
odd, since $n=8l=\left( 4d+4\right) $. The matrices of order $8l$
constructed in with this method will be denoted by $M_{4,4,l}$. Essentially,
orthogonality arises since we \emph{interlace} $M_{4,4}$ with itself $d+1$
times. This can be seen directly in the matrix%
\begin{eqnarray}
M_{4,4,1}=AW(8,4)= \nonumber \\
\left( 
\begin{array}{rrrrrrrr}
0 & 1 & 0 & 1 & 0 & 1 & 0 & -1 \\ 
\fbox{$1$} & 0 & \fbox{$1$} & 0 & \fbox{$1$} & 0 & \fbox{$-1$} & 0 \\ 
0 & 1 & 0 & 1 & 0 & -1 & 0 & 1 \\ 
\fbox{$1$} & 0 & \fbox{$1$} & 0 & \fbox{$-1$} & 0 & \fbox{$1$} & 0 \\ 
0 & 1 & 0 & -1 & 0 & 1 & 0 & 1 \\ 
\fbox{$1$} & 0 & \fbox{$-1$} & 0 & \fbox{$1$} & 0 & \fbox{$1$} & 0 \\ 
0 & -1 & 0 & 1 & 0 & 1 & 0 & 1 \\ 
\fbox{$-1$} & 0 & \fbox{$1$} & 0 & \fbox{$1$} & 0 & \fbox{$1$} & 0%
\end{array}%
\right) ,  \label{AW(8,4)}
\end{eqnarray}%
which is constructed by interlacing two copies of $M_{4,4}$. The boxed
numbers are the nonzero entries in one of the two copies. At this stage, we
can ask information about the graphs of matrices in $\mathcal{M}_{4,4}$. It
is sufficient to observe that $M_{4,4,l}=M_{4,4}\otimes F_{l+1}$, where the
matrix%
\[
F_{n} = \left( 
\begin{array}{cccc}
0 & \cdots  & 0 & 1 \\ 
\vdots  &  & 1 & 0 \\ 
0 &  &  & \vdots  \\ 
1 & 0 & \cdots  & 0%
\end{array}%
\right) ,
\]%
is $n\times n$.  Note that this is just the matrix $P_\pi$ from Proposition \ref{Reduction}, but we have changed notation to make the dependence on the size of the matrix explicit. The matrix $F_{n}$ is the adjacency matrix of the disjoint
union of $n/2$ graphs $K_{2}$, if $n$ is even and the disjoin union of $(n-1)/2$
graphs $K_{2}$ and a single vertex with a self-loop, if $n$ is odd. The 
\emph{tensor product} $G\otimes H$ of graphs $G$ and $H$ is the graph such
that $A(G\otimes H)=A(G)\otimes A(H)$. The set of vertices of $G\otimes H$
is the Cartesian product $V(G)\times V(H)$ and two vertices $\{u,u^{\prime
}\}$ and $\{v,v^{\prime }\}$ are adjacent in $G\otimes H$ if and only if $%
\{u,v\}\in E(G)$ and $\{u^{\prime },v^{\prime }\}\in E(H)$. The graph $%
G\otimes H$ is connected if and only if both factors are connected and at
least one factor is nonbipartite. Here, $K_{2}$ and $K_{4}$ are both
connected and $K_{2}$ is bipartite. When $l>1$, one of the factors is not
connected. Therefore $G(M_{4,4,l})$ is connected only if $l=1$. In fact $%
G(M_{4,4,l})\cong \biguplus\nolimits_{\#l}K_{4,4}$, where $K_{n,n}$ denotes
the complete bipartite graph on $2n$ vertices. In Figure \ref{K44} is
illustrated $G(M_{4,4,1})$; in Figure \ref{2K44}, $G(M_{4,4,4})$.

\begin{figure}
\begin{center}
\includegraphics[height=4cm]{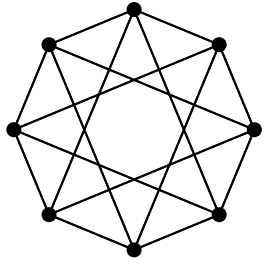}
\caption{The graph $G(M_{4,4,1})\protect\cong K_{4,4}$}
\label{K44}
\end{center}
\end{figure}

\begin{figure}
\begin{center}
\includegraphics[height=4cm]{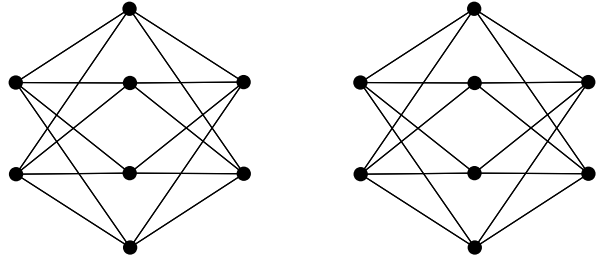}
\caption{The graph $G(M_{4,4,4})\protect\cong K_{4,4}\uplus K_{4,4}$.}
\label{2K44}
\end{center}
\end{figure}

The graph $G(M_{4,4,l})$ is made of $l$ connected components. Can
we find two permutation matrices $P$ and $Q$ such that $G(PM_{4,4,l}Q)\ncong
G(M_{4,4,l})$ and it is connected? In particular, for the moment, we
consider only permutation matrices that preserves the distance $d$ between
the elements of $S(M)$. From $M_{4,4,l}$, we can construct exactly $l-1$
further matrices, by permutation the lines under this constraint. For
example, when $n=24=8\cdot 3$, the possible first rows are%
\[
\begin{tabular}{l}
$0,0,0,0,0,1,0,0,0,0,0,1,0,0,0,0,0,1,0,0,0,0,0,-1$ \\ 
$0,0,0,1,0,0,0,0,0,1,0,0,0,0,0,1,0,0,0,0,0,-1,0,0$ \\ 
$0,1,0,0,0,0,0,1,0,0,0,0,0,1,0,0,0,0,0,-1,0,0,0,0$%
\end{tabular}%
.
\]%
The position of the entry $-1$ does not matter, because we are interested in
the graph $G(M_{4,4,l})$. Given $s_{1},...,s_{4}$, in matrices of the form $%
M_{4,4,l}$, when $l>1$,%
\[
\gcd (\{s_{i}-s_{j}:\mbox{for all }i,j\mbox{ such that }j>i\}\cup 8l)>2.
\]%
This implies that $G(PM_{4,4,l}Q)\cong \biguplus\nolimits_{\#l}K_{4,4}$, for
any two permutations $P$ and $Q$ preserving the distance $d$.

%------------------------------------------------------------------------------------------------------------%
\subsection{Graphs from \texorpdfstring{$AW(6,4)$}{AW(6,4)}}\label{AW(6,4)}
%------------------------------------------------------------------------------------------------------------%

It is natural to analyse the construction in the previous section, but
replacing the matrix $AW(4,4)$ with some other anticirculant weighing
matrix. Recalling that there is no $AW(5,4)$, the smallest available matrix
of order $n>4$ is $AW(6,4)$:%
\[
M_{6,4}=AW(6,4)=\left( 
\begin{array}{rrrrrr}
-1 & 1 & 0 & 1 & 1 & 0 \\ 
1 & 0 & 1 & 1 & 0 & -1 \\ 
0 & 1 & 1 & 0 & -1 & 1 \\ 
1 & 1 & 0 & -1 & 1 & 0 \\ 
1 & 0 & -1 & 1 & 0 & 1 \\ 
0 & -1 & 1 & 0 & 1 & 1%
\end{array}%
\right) .
\]%
This is obtained from the $CW(6,4)$ in \eref{ci6}. The graph $G(M_{6,4})
$, illustrated in Figure \ref{W64}, has $4$ self-loops. 

\begin{figure}
\begin{center}
\includegraphics[height=4cm]{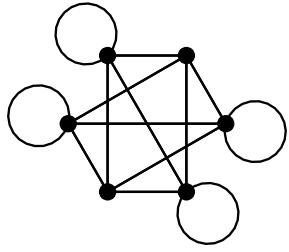}
\caption{The graph $G(M_{6,4})$.}
\label{W64}
\end{center}
\end{figure}

The graph $G(M_{6,4,l})$, if loopless, is
on $12l$ vertices, with $l\geq 1$. The graph $G(M_{6,4,1})\cong
G(M_{6,4})\otimes K_{2}$ is connected since $K_{2}$ and $G(M_{6,4})$ are
connected and $G(M_{6,4})$ is nonbipartite. (The chromatic number of $%
G(M_{6,4})$ is $3$.) A drawing of $G(M_{6,4,1})$ is in Figure \ref{G641}. 

\begin{figure}
\begin{center}
\includegraphics[height=4cm]{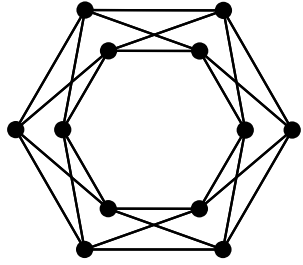}
\caption{The graph $G(M_{6,4,1})\protect\cong G(M_{6,4})\otimes K_{2}$.}
\label{G641}
\end{center}
\end{figure}

In general, $M_{6,4,l}=M_{6,4}\otimes F_{l+1}$. If we avoid
self-loops, 
\[
G(M_{6,4,l})=\biguplus\nolimits_{l}G(M_{6,4,1}).
\]%
By permuting the lines of $M_{6,4}$ (\emph{shifting to the right} the
columns), we have the matrix%
\[
M_{6,4}^{1}=AW(4,4)_{1}=\left( 
\begin{array}{cccccc}
0 & -1 & 1 & 0 & 1 & 1 \\ 
-1 & 1 & 0 & 1 & 1 & 0 \\ 
1 & 0 & 1 & 1 & 0 & -1 \\ 
0 & 1 & 1 & 0 & -1 & 1 \\ 
1 & 1 & 0 & -1 & 1 & 0 \\ 
1 & 0 & -1 & 1 & 0 & 1%
\end{array}%
\right) .
\]%
Apply our construction to $M_{6,4}^{1}$, we obtain graphs on $6k$ vertices,
with $k$ odd. However, since $G(M_{6,4,1})\cong G(M_{6,4,1}^{1})$, these
graphs will be a disjoint union of copies of $G(M_{6,4,1})$. So the
permutation is of no use if our intention is to get different graphs from $%
M_{6,4}$. On the other side, the structure of $M_{6,4}^{1}$ suggests a
different construction, which will allow us to obtain matrices of the form $%
AW(4m,4)$, whose graph is loopless and connected for every $m\geq 3$. Notice
that%
\[
M_{6,4}=\left( 
\begin{array}{cc}
A & B \\ 
B & A%
\end{array}%
\right) ,
\]%
where 
\[
\begin{tabular}{lll}
$A=\left( 
\begin{array}{rrr}
0 & -1 & 1 \\ 
-1 & 1 & 0 \\ 
1 & 0 & 1%
\end{array}%
\right) $ &  & $B=\left( 
\begin{array}{rrr}
0 & 1 & 1 \\ 
1 & 1 & 0 \\ 
1 & 0 & -1%
\end{array}%
\right) .$%
\end{tabular}%
\]%
Then $2AB=0$ and $A^2+B^2=4I$ and indeed
\[
M_{6,4}M_{6,4}=\left( M_{6,4}\right) ^{2}=\left( 
\begin{array}{cc}
A^{2}+B^{2} & 2AB \\ 
2AB & A^{2}+B^{2}%
\end{array}%
\right) =4I.
\]%
We can add an even number of extra lines to $M_{6,4}$ and still preserve its
block structure. For example, the matrices%
\[
\hspace{-20pt}
A^{\prime }=\left( 
\begin{array}{rrrrrr}
0 & -1 & 0 & 1 & 0 & 0 \\ 
-1 & 0 & 1 & 0 & 0 & 0 \\ 
0 & 1 & 0 & 0 & 0 & 1 \\ 
1 & 0 & 0 & 0 & 1 & 0 \\ 
0 & 0 & 0 & 1 & 0 & 1 \\ 
0 & 0 & 1 & 0 & 1 & 0
\end{array}\right) \ , \
B^{\prime }=\left( 
\begin{array}{cccccc}
0 & 1 & 0 & 1 & 0 & 0 \\ 
1 & 0 & 1 & 0 & 0 & 0 \\ 
0 & 1 & 0 & 0 & 0 & -1 \\ 
1 & 0 & 0 & 0 & -1 & 0 \\ 
0 & 0 & 0 & -1 & 0 & 1 \\ 
0 & 0 & -1 & 0 & 1 & 0
\end{array}\right) 
\]
satisfy the above conditions. Hence the matrix

\[
M_{12,4}=\left( 
\begin{array}{cc}
A^{\prime } & B^{\prime } \\ 
B^{\prime } & A^{\prime }
\end{array}
\right) 
\]
is an $AW(12,4)$. In general, let $\mathcal{M}_{6,4}$ be the family of
matrices obtained with this construction. If $M\in \mathcal{M}_{6,4}$ then $%
s_{1}=2$, $s_{2}=4$, $s_{3}=n/2+2$ and $s_{4}=n/2+4$, where $%
S(M)=\{s_{1},s_{2},s_{3},s_{4}\}$. These are $AW(4k,4)$ when $%
s_{3}-s_{2}=\left( n-6\right) /2$ is an odd number. The construction gives
also matrices of order $2l\geq 6$, for any $l$. Yet the associated graphs
will have loops when $l$ is odd. This is the reason for taking multiples of $%
4$. If $M\in \mathcal{M}_{6,4}$ then $G(M)$ is connected, in virtue of the
fact that $s_{2}-s_{1}=4-2=2$ and then $\gcd (\{s_{i}-s_{j}:$ for all $i,j$
such that $j>i\}\cup 4k\}=2$.

The graphs of matrices in the set $\mathcal{M}_{6,4}$, with order $4m$, are the direct product of $C_{2m}$ with the graph $K_{2}^{+}$. (Recall that $K_{n}^{+}$ is the complete graph on $n$ vertices with self-loops; these are the graphs with adjacency matrix $J_{n}$.) The \emph{direct product }$G\times H$ of graphs $G$ and $H$ has set of vertices $V(G\times H)=V(G)\times V(H)$ and two vertices $\{u,u^{\prime }\}$ and $\{v,v^{\prime }\}$ are adjacent in $G\times H$ when $\{u,v\}\in E(G)$ and $\{u^{\prime },v^{\prime }\}\in E(H)$. By the definition, $A(G\times H)=A(G)\otimes A(H)$. For this reason $G\times H$ is sometimes called the \emph{Kronecker product of graphs} (see Imrich and Klav\v{z}ar \cite{Imrich2000}, Ch. 5). If $G\in \mathcal{M}_{6,4}$ then $G\cong K_{2}^{+}\times C_{2m}$ and $M(G)=J_{2}\otimes A(C_{2m})$. 

%------------------------------------------------------------------------------------------------------------%
\subsection{Graphs from \texorpdfstring{$AW(7,4)$}{AW(7,4)}}\label{AW(7,4)}
%------------------------------------------------------------------------------------------------------------%

Let us consider the matrix
\[
M_{7,4}=AW(7,4)=\left( 
\begin{array}{rrrrrrr}
0 & 1 & 0 & 1 & 1 & -1 & 0 \\ 
1 & 0 & 1 & 1 & -1 & 0 & 0 \\ 
0 & 1 & 1 & -1 & 0 & 0 & 1 \\ 
1 & 1 & -1 & 0 & 0 & 1 & 0 \\ 
1 & -1 & 0 & 0 & 1 & 0 & 1 \\ 
-1 & 0 & 0 & 1 & 0 & 1 & 1 \\ 
0 & 0 & 1 & 0 & 1 & 1 & -1
\end{array}
\right) .
\]
The graph of $M_{7,4}$ is in Figure \ref{7aw}.

\begin{figure}
\begin{center}
\includegraphics[height=4cm]{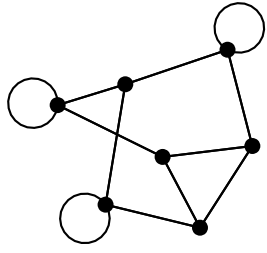}
\caption{The graph $G(M_{7,4})$.}
\label{7aw}
\end{center}
\end{figure}

\begin{figure}
\begin{center}
\includegraphics[height=5cm]{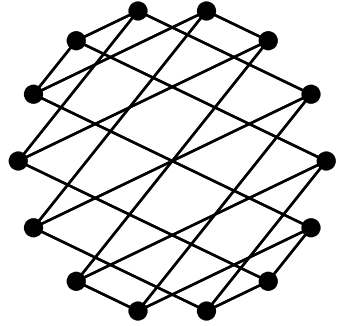}
\caption{The graph $G(M_{7,4,1})$.}
\label{aw14}
\end{center}
\end{figure}

Interlacing $M_{7,4}$ with itself, by the same method used for $M_{4,4}$ (see \sref{AW(4,4)}, \eqref{AW(8,4)}), we obtain the matrix 
\begin{equation}
M_{7,4,1} = AW(14,4) = AW(7,4) \otimes F_2
\end{equation}
The graph of $M_{7,4,1}$ is in Figure \ref{aw14}.

This is on $14$ vertices, $4$-regular and connected. At a first analysis, $G(M_{7,4,1})$ does not seem to be a direct product of graphs. It is nontrivial to extend $G(M_{7,4,1})$ to an infinite family of connected regular graphs, as we did by taking $M_{6,4}$. Examples show that graphs obtained by $M_{7,4}$ are not connected, but $G(M_{7,4,1})$. 

%------------------------------------------------------------------------------------------------------------%
\section{Open problems}\label{S:openproblems}
%------------------------------------------------------------------------------------------------------------%

\begin{itemize}
\item Study the structure of the graph $G(M_{7,4,1})$. Determine if it belongs to a family of connected graphs arising from $AW(7,4)$, having an infinite number of members, as it is for the graphs arising from $AW(6,4)$. 

\item Since there is no general classification of weighing matrices, characterizing graphs of Hankel weighing matrices seems to be out of reach. This could be however done for graphs of degree $4$ and $9$, on the basis of present knowledge on $CW(n,4)$ and $CW(n,9)$ (see \cite{Arasu1998}).

\item We have seen that there are examples of Hankel weighing matrices which are not anticirculant. A way to study the relation between these matrices and anticirculant ones, would be to prove that the graphs of Hankel weighing matrices are always isomorphic to graphs of $AW(n,d)$, except for some special cases which will include the examples in proposition \ref{Inequivalence}.

\begin{figure}
\begin{center}
\includegraphics[height=3.5cm]{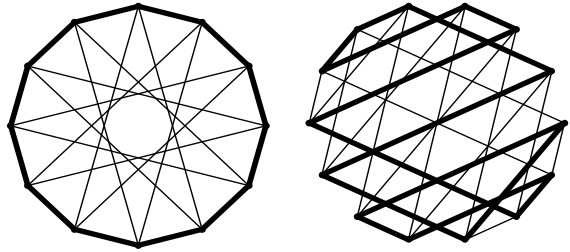}
\caption{Hamiltonian cycles in graphs of matrices $AW(12,4)$ and $AW(14,4)$.}
\label{hamiltoniangraphs}
\end{center}
\end{figure}

\item A \emph{Hamiltonian cycle} in a graph is an ordered set of sequentially adjacent vertices, in which every vertex of the graph appears exactly once. Graphs with Hamiltonian cycles are said to be \emph{Hamiltonian}. Alspach and Zhang \cite{Alspach1989} conjectured that a connected Cayley graph $G(D_{n},T)$ is Hamiltonian whenever $T$ is a set of reflections. According to this conjecture, any connected anticirculant graph is Hamiltonian whenever it is loopless. Additionally, Gutin \emph{et al. } \cite{Gutin2006} conjectured that graphs of unitary matrices are Hamiltonian if connected. According to these conjectures the graphs of Hankel weighing matrices are Hamiltonian if connected. This is the case for the graphs considered here. In Figure \ref{hamiltoniangraphs} are drawn the graphs $K_{2}^{+}\times C_{8}$ and $G(M_{7,4,1})$. The highlighted edges represent Hamiltonian cycles. Perhaps the physical picture that motivates our study could provide some insight toward proving these conjectures.

\item We have seen that graphs of the form $K_{2}^{+}\times C_{2m}$ are graphs of unitary matrices. The quantum dynamics induced by these matrices is trivial because of symmetry. However symmetry can be broken with the action of a diagonal matrix without altering the zero pattern of the unitary. This suggests the possibility of constructing unitaries with a nontrivial dynamics whose graph is a undirected Cayley graph of the dihedral group.

\item Sections \ref{S:2reg} and \ref{S:4reg} make extensive use of the following trick.  Begin with a circulant weighing matrix, then reverse the columns.  If the resulting graph has loops, one can create a loopless graph by the ``interlacing'' trick.  In fact the interlaced graph will again be connected if the original graph was connected but not bipartite.  Can one classify the Hankel and anticirculant weighing matrices which cannot be obtained by this trick?

\end{itemize}

% -----------------------------------------------------------------------------------------------------------%
\section{Conclusion}\label{S:conc}
% -----------------------------------------------------------------------------------------------------------%

We have demonstrated connections between an approach to quantum computation with optical modes and the existence of certain weighing matrices with Hankel structure.  Because the degree of the graph associated to a matrix corresponds to the experimental resources required to implement the graph as a CVCS, it is important to characterize which graphs with small degree have adjacency matrices with Hankel structure.  We proved some general theorems about such matrices, classified completely the case of degree 2, and provided some examples with degree 4.  We also raised some interesting open problems that might help shed some light on which graphs can be implemented in a single OPO using the scheme in Ref.~\cite{Menicucci2008,Flammia2008}.

% -----------------------------------------------------------------------------------------------------------%
\section*{Acknowledgments}%\ack
% -----------------------------------------------------------------------------------------------------------%

STF was supported by the Perimeter Institute for Theoretical Physics.  Research at Perimeter is supported by the Government of Canada through Industry Canada and by the Province of Ontario through the Ministry of Research~\& Innovation.  SS was supported by the Institute for Quantum Computing.  Research at the Institute for Quantum Computing is supported by DTOARO, ORDCF, CFI, CIFAR, and MITACS.

% -----------------------------------------------------------------------------------------------------------%
\section*{References}
% -----------------------------------------------------------------------------------------------------------%

\bibliography{weighing}
\bibliographystyle{unsrt}

% -----------------------------------------------------------------------------------------------------------%

\end{document}